\documentclass[conference,letterpaper]{IEEEtran}

\usepackage{cite}
\usepackage{subfig}
\usepackage[utf8]{inputenc} 
\usepackage[T1]{fontenc}
\usepackage{url}
\usepackage{ifthen}
\usepackage{cite}
\usepackage[cmex10]{amsmath} 

\usepackage{graphicx,multicol,amsthm} 
\usepackage{epsfig} 
\usepackage{amsmath} 
\usepackage{amssymb}  
\usepackage{mathtools}
\usepackage{float}
\usepackage{enumitem}
\usepackage{cases,setspace,adjustbox,xspace}

\usepackage{texdef2015}
\usepackage{cite}
\usepackage{algorithmic}
\usepackage{array}
\usepackage{makecell}
\newcolumntype{L}{r>{\centering\arraybackslash}m{.5cm}}

\usepackage{amstext} 
\usepackage[update,prepend]{epstopdf}
\usepackage{fixltx2e}
\usepackage{graphicx,multicol,amsthm} 
\usepackage{epsfig} 

\newcommand{\ageM}{\Delta}
\newcommand{\ageMAvg}{\overline{\Delta}}
\newcommand{\ageGW}{\Delta'}

\newcommand{\ageGWVec}{\mathbf{\ageGW}}
\newcommand{\ageMVec}{\mathbf{\ageM}}

\newcommand{\ageGWVecCoupled}{\mathbf{\widetilde{\ageGW}}}
\newcommand{\ageMVecCoupled}{\mathbf{\widetilde{\ageM}}}

\newcommand{\ageMCoupled}{\widetilde{\ageM}}
\newcommand{\ageGWCoupled}{\widetilde{\ageGW}}

\DeclareMathOperator*{\argmin}{arg\,min}
\begin{document}
\title{Minimizing Age in Gateway Based Update Systems}


\IEEEoverridecommandlockouts
\author{%
  \IEEEauthorblockN{Sandeep Banik$^{\dagger}$, Sanjit K. Kaul$^{\dagger}$$^{*}$ and P. B. Sujit$^{\dagger}$}
  \IEEEauthorblockA{$^{*}$Wireless Systems Lab, IIIT-Delhi, $^{\dagger}$Department of ECE, IIIT-Delhi\\
                    Email: \{sandeepb, skkaul, sujit\}@iiitd.ac.in}\thanks{This work was partially supported by the EPSRC GCRF grant EP/P02839X/1 and the Young Faculty Research Fellowship (Visvesvaraya PHD scheme) received by Sanjit Kaul. We acknowledge the support of the Infosys Center for Artificial Intelligence at IIIT-Delhi.}
}


\maketitle

\begin{abstract}
We consider a network of status updating sensors whose updates are collected and sent to a monitor by a gateway. The monitor desires as fresh as possible updates from the network of sensors. The gateway may either poll a sensor for its status update or it may transmit collected sensor updates to the monitor. We derive the average age at the monitor for such a setting. We observe that increasing the frequency of transmissions to the monitor has the upside of resetting sensor age at the monitor to smaller values. However, it increases the length of time that elapses before a sensor is polled again. This motivates our investigation of policies that fix the number of sensors s the gateway polls before transmitting to the monitor. 

For any s, we show that when sensor transmission times to the gateway are independent and identically distributed (iid), for independent but possibly non-identical transmission times to the monitor, it is optimal to poll a sensor with the maximum age at the gateway first. Also, under simplifying assumptions, the optimal value of s increases as the square root of the number of sensors. For non-identical sensor transmission times, we consider a policy that polls a sensor such that the resulting average change in age is minimized. We compare our policy proposals with other policies, over a wide selection of transmission time distributions.
\end{abstract}

\section{Introduction}
Internet-of-Things (IoT) deployments often have	large numbers of devices send their information to a monitor (an Internet server or a monitoring station) via a gateway. In general, the gateway could be a wireless access point or could be a mobile node like a UAV. The gateway must collect updates from the IoT devices (sensors) and deliver them to the monitor that desires sensed information to be fresh at its end. In such settings, the gateway must decide between updating the monitor with currently obtained sensed information and obtaining newer sensed information before sending all collected information to the monitor. When the gateway is sending information to the monitor, newer sensor updates must wait for the sending to end before they can hope to be delivered to the monitor. Update from a sensor may also have to wait for updates from others to be collected before all are sent to the monitor together. 

In this work, we investigate a setting where the gateway can either poll a sensor for a fresh update or send collected updates to the monitor. On being polled, a sensor begins transmission to the gateway, which takes a random amount of time. The gateway must wait for the transmission to end before it can make its next decision. The same applies for when the gateway chooses to transmit collected updates to the monitor. Figure~\ref{fig:gatewayIllustration} shows an illustration. We derive the average age of sensor updates at the monitor and observe that frequent transmissions to the gateway reset the age of sensed information at the monitor to smaller values. However, this comes at the cost of larger intervals of time between a given sensor being polled.
\begin{figure}[t]
	\centering
	\includegraphics[width=.2\textheight]{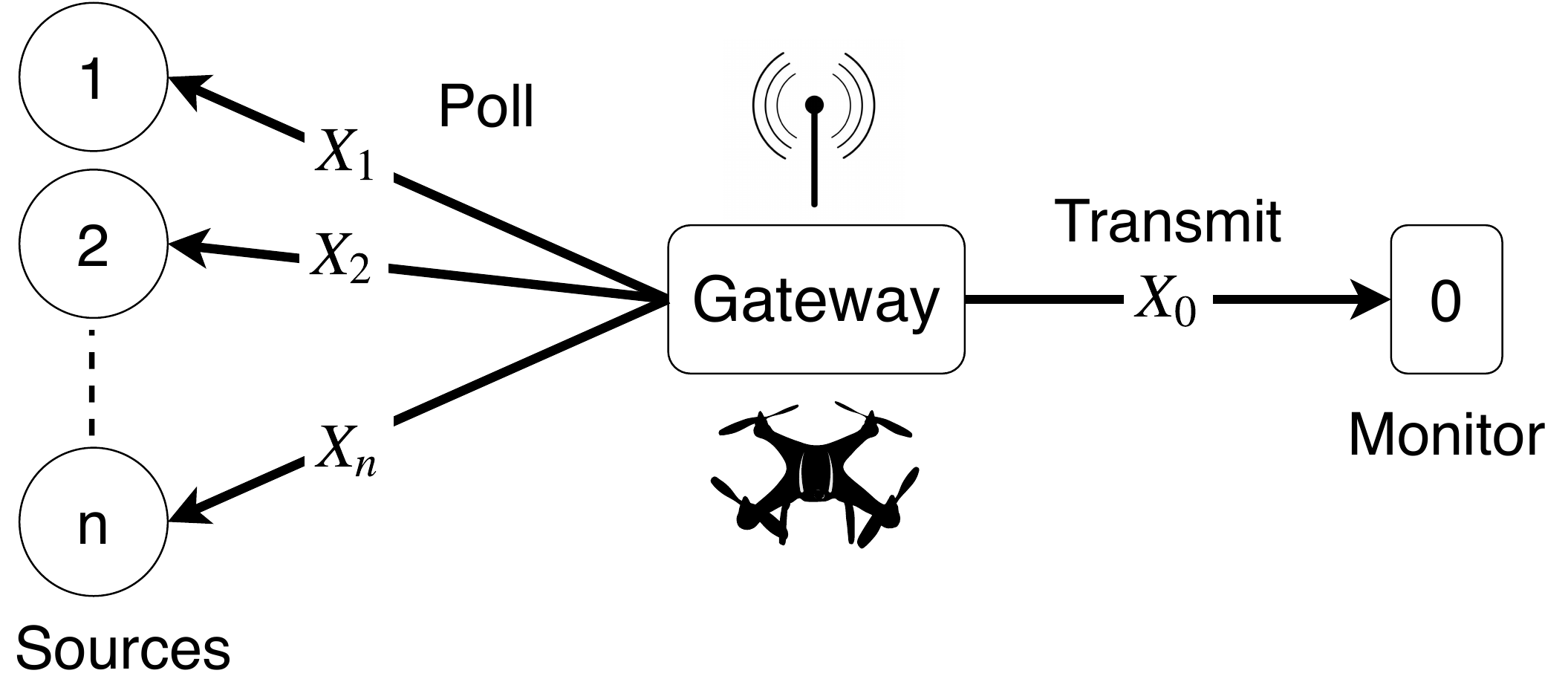}
	\caption{\small We have $n$ sensors, a gateway, and a monitor. Sensor $i$'s transmission time is distributed as $f_{X_i}(x)$. Transmission to the monitor is distributed as $f_{X_0}(x)$.}
	\label{fig:gatewayIllustration}
\vspace{-0.2in}	
\end{figure}

To capture the trade-off while keeping the decision problem at the gateway tractable, we investigate a class of policies that force the gateway to transmit collected sensor updates to the gateway after a fixed $s$ number of sensors have been polled. It turns out that for these policies, for any given $s$, for when the sensor transmission times are iid and the transmission to the monitor is independent but may be non-identically distributed, a policy that polls a sensor that has the maximum age at the \emph{gateway} first minimizes the average age at the \emph{monitor}. We compare such a policy with an optimal choice of $s^*$ with differing choices of polling sensors and frequencies of transmissions to the monitor. We also observe that under certain assumptions the optimal $s^*$ increases as the square root of the number of sensors. Finally, we devise and evaluate a policy, along the lines of the maximum age first, for when the sensor transmissions may also be non-identically distributed.

Various recent works~\cite{AoI-Igor-SPBN,AoI-BN-RA,SA-BN-nb,OLS-MM-PAOI, OFI-LS, AO-STS-MSS} have looked at the problem of scheduling in wireless networks. The authors in
\cite{AoI-Igor-SPBN,AoI-BN-RA,SA-BN-nb} considered the problem of minimizing the AoI in a wireless broadcast network. The maximum peak age of information is minimized using a fast scheduling algorithm in \cite{OLS-MM-PAOI}. In \cite{OFI-LS} a number of links with a common channel is considered and a scheduling algorithm is developed capable of handling scalability of the network. Minimizing the AoI of a multi-source system has been studied by \cite{AO-STS-MSS}, where only one source can be scheduled to the destination node via a channel.


The model that our work considers is different given the presence of a gateway that must aggregate sensor information and then send the same to a monitor. The maximum age first policy being optimal has been shown in different settings, for example, in~\cite{AoI-Igor-SPBN} and~\cite{AO-STS-MSS}. One can draw parallels between our work and these settings under suitable assumptions including that the transmissions to the monitor are instantaneous. As a result, we benefit from some of the proof techniques in these works. There are similarities with the model in~\cite{TFM-LCF} in that the work has broadly two classes of updating processes, one that has sources update information at their end and the other that has a server keep itself updated with information at the sources. In~\cite{AO-TP-UAV} the authors consider planning a UAV's path that takes it through sensors and back to the base station (monitor). However, they restrict themselves to plans that visit every sensor once before going to the monitor. Given our takeaways, we believe that there should be path plans where the UAV visits the monitor more frequently that lead to a smaller age. Finally, in~\cite{TripathiAgeGatheringDissemGraphs} the authors consider information dissemination from ground terminals (sensors) to a central terminal (monitor) using a mobile agent. However, communication between the mobile agent and the terminal is instantaneous.


The paper is organized as follows. We formulate the problem in Section~\ref{sec:model} and derive the average age at the monitor in Section~\ref{sec:avgAge}. In Section~\ref{sec:policy}, we describe the above mentioned policy, state its properties in Theorem~\ref{thm:MAF}. We also consider the case when sensor transmission times are non-identical. We conclude with a summary evaluation in Section~\ref{sec:eval}.

\section{Problem Formulation}
\label{sec:model}
Our network (see Figure~\ref{fig:gatewayIllustration}) consists of $n$ sensors indexed $1,2,\ldots,n$, a gateway, and a monitor indexed $0$. The gateway polls sensors for their fresh updates and transmits the same to the monitor. On being polled, a sensor $i$ transmits a fresh update to the gateway. The transmission of an update by sensor $i \ge 1$ takes a random amount of time $X_i$ distributed as $f_{X_i}(x)$, independently of transmission times of other updates sent by it, other sensors, or the gateway. Time taken by a transmission from the gateway to the monitor is distributed as $f_{X_0}(x)$ and is independent of other transmission times. In practice, random transmission times may result from packet retries carried out by, for example, the wireless link layer.

In this work, we will impose the constraint that exactly one transmission can take place at any given time. In practice, this may be desirable when the sensors and the gateway transmit over a shared wireless spectrum and interference due to overlapping transmissions cause them to be received in error. The gateway must decide between either polling a sensor and receiving the resulting transmission of the sensor's update or transmitting to the monitor polled updates that have not yet been transmitted to it. In either case, it must wait for the ensuing transmission to end before making its next decision.

Let $b_k, k\ge 1$ be the time the monitor receives the $k$\textsuperscript{th} transmission from the gateway. Let $b_0 = 0$. This transmission will include the status updates of all the sensors that were polled by the gateway starting at time $b_{k-1}$. Suppose a total of $s_k$ sensors indexed $i_1, i_2,\ldots, i_{s_k}$ were polled in the stated sequence. The interval $(b_{k-1}, b_k)$ is of random length. It starts with a transmission of length $X_{i_1}$ distributed as $f_{X_{i_1}}$ followed by transmissions by the $s_k - 1$ other sensors in the above specified order and lastly a transmission of length $X_0$ distributed as $f_{X_0}$ by the gateway to the monitor. As a result $b_k = b_{k-1} + X_0 + \sum_{j=1}^{s_k} X_{i_j}$.

Let $\ageM_{i}(t)$ be the age of updates from sensor $i$ at the monitor. If $u_i(t)$ is the timestamp of the most recent update of $i$ at the monitor, then $\ageM_{i}(t) = t - u_i(t)$. Correspondingly, let $\ageGW_{i}(t)$ be the age of status updates from sensor $i$ at the gateway. We assume that $\ageM_{i}(0) = \ageGW_{i}(0)$ for all sensors $i$. Note that at all time instants $b_k$, $k\ge 1$, the ages of the sensors at the monitor must be reset to their ages at the gateway. 

Now consider how the age processes evolve over the interval $(b_{k-1}, b_{k})$. As before, assume that a total of $s_k$ sensors were polled before the gateway chose to send its $k$\textsuperscript{th} transmission to the monitor. The $j$\textsuperscript{th} sensor polled is $i_j$. Let $\tau_{j}, j\ge 1$ be the time when the transmission by sensor $i_j$ to the gateway completes. Let $\tau_0 = b_{k-1}$. We have, for $j\in \{1,\ldots,s_k\}$ and any sensor $i$,
\begin{align}
\ageGW_{i}(\tau_{j}) &=
\begin{cases}
 X_{i_j} & i = i_j,\\
\ageGW_{i}(\tau_{j-1}) + X_{i_j} & i\ne i_j.
\end{cases}
\label{eqn:ageEvolAtGWa}
\end{align}
Further, at $b_k$, we have
\begin{align}
\ageGW_{i}(b_k) &= \ageGW_{i}(\tau_{s_k}) + X_0.
\label{eqn:ageEvolAtGWb}
\end{align}

On the other hand, the age $\ageM_{i}(t)$ of sensor $i$ at the monitor increases while the gateway receives transmissions from sensors. At time $b_k$, when the monitor has received the transmission from the gateway, the age $\ageM_{i}(b_k)$ of any sensor $i$ at the monitor is reset to $\ageGW_{i}(b_k)$. Note that sensors that were not polled during $(b_{k-1}, b_{k})$ will not see any reduction in age due to this reset. We have, for any sensor $i$,
\begin{align}
\ageM_{i}(b_k) &= \ageGW_{i}(b_k).
\label{eqn:ageEvolAtMonitor}
\end{align}


Define the vector of ages at the monitor as $\ageMVec(t) = \vec{\ageM_1(t)&\cdots&\ageM_n(t)}$. Let the corresponding vector at the gateway be $\ageGWVec(t) = \vec{\ageGW_1(t)&\cdots&\ageGW_n(t)}$. At any decision time $t$, the choice of polling a certain sensor or, instead, sending to the monitor is made as a function of the age vectors $\ageMVec(t)$ and $\ageGWVec(t)$. Note that decision times include $t=0$ and all times that mark the end of transmission that ensued as a result of the preceding decision. Define the set of all decisions $\mathcal{A} = \{0,1,\ldots,n\}$. 

Let $\pi$ be a policy that maps $\vec{\ageMVec(t), \ageGWVec(t)}$ to decisions in the set $\mathcal{A}$. The monitor would like to have as fresh as possible sensed information from all the sensors. Note that the sample function of age $\ageM_i(t)$ of sensor $i$ is a function of policy $\pi$. The time-average age $\ageMAvg_i^\pi$ of sensor $i$ at the monitor is 
\begin{align}
\ageMAvg_i^\pi = \lim_{t\to \infty} \frac{1}{t} \int_{0}^{t}\ageM_i(t).
\label{eqn:avgAgei}
\end{align}

Define the corresponding age-of-information (AoI) of the network of sensors at the monitor as 
\begin{align}
\ageMAvg^\pi = \frac{1}{n} E\left[\sum_{i=1}^{n} \ageMAvg_i^\pi\right].
\label{eqn:AoI}
\end{align}

We will restrict ourselves to policies that don't use any knowledge of the future when making a decision. Let the set of such policies be $\Pi$. We would like to find a policy $\pi^*\in \Pi$ that minimizes the AoI at the monitor. We want to find $\pi^* = \argmin_{\pi \in \Pi} \ageMAvg^\pi$. %
To simplify notation, we skip the explicit mention of $\pi$.
\section{Average Age of a Sensor}
\label{sec:avgAge}
\begin{figure}[t]
	\includegraphics[scale=0.5]{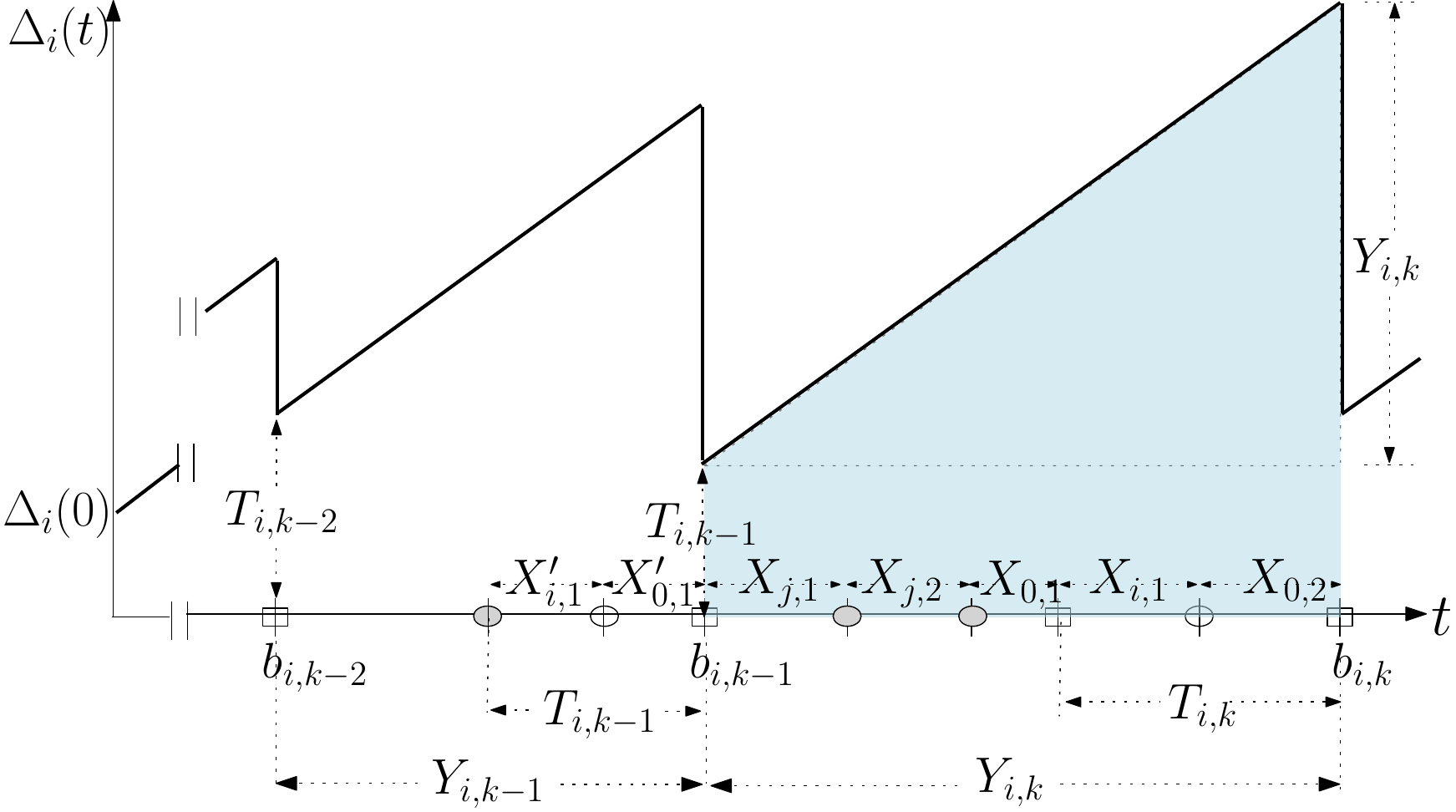}
	\caption{\small Sample function of age of sensor $i$ at the monitor. The empty circles are time instants when a fresh update from $i$ was received by the gateway. The filled circles correspond to the same for another sensor $j$. The empty squares are when the monitor receives a transmission from the gateway.}
	\label{fig:age}
\vspace{-0.1in}	
\end{figure}
We will derive the expression for the average age $\ageMAvg_i^\pi$ of sensor $i$ at the monitor. Let $b_{i,k}, k\ge 1$, be the time the monitor receives the $k$\textsuperscript{th} such transmission from the gateway that resets the age of sensor $i$ at the monitor. It is the $k$\textsuperscript{th} transmission that contains a newly polled status update from sensor $i$. Let $Y_{i,k}$ be the time between the $k-1$\textsuperscript{th} and $k$\textsuperscript{th} such reset of age of $i$ at the monitor. We have $Y_{i,k} = b_{i,k} - b_{i,k-1}$.

Let $\kappa_{i,k}$ be the time of most recent polling of sensor $i$ before $b_{i,k}$. This polling time will lie in the interval $(b_{i,k-1}, b_{i,k})$. Let $T_{i,k} = b_{i,k} - \kappa_{i,k}$. $T_{i,k}$ is the age of the update from the sensor $i$ when it is received by the monitor at $b_{i,k}$. The age $\ageM_i(t)$ resets to $T_{i,k}$ at $t= b_{i,k}$. We have $\ageM_i(b_{i,k}) = T_{i,k}$.

As illustrated in Figure~\ref{fig:age}, the area under the sample function $\ageM_i(t)$ can be written as a sum of areas under the function over the intervals $(b_{i,k-1}, b_{i,k}), k\ge 1$. The area in a given interval $(b_{i,k-1}, b_{i,k})$ is simply $Q_{i,k} \triangleq Y_{i,k} T_{i,k-1} + Y^2_{i,k}/2$. Using this fact and Equation~(\ref{eqn:avgAgei}) we can write the average age $\ageMAvg_i$ as
\begin{align}
\ageMAvg_i &= \lim_{t \to \infty} \frac{k(t)}{t} \frac{1}{k(t)} \sum_{j=1}^{k(t)} \left[Y_{i,j} T_{i,j-1} + \frac{Y^2_{i,j}}{2}\right],\nonumber
\end{align}
where $k(t)$ is the number of such areas that span the interval $(0,t)$. Assuming in the limit as $t\to \infty$, $\frac{k(t)}{t} \to 1/E[Y_i]$ and $(Q_{i,k}/k) \to E[Q_i]$, $Y_{i,j} \sim Y_i$ and $T_{i,j-1} \sim T_i$, we can write
\begin{align}
\ageMAvg_i &= \frac{1}{E[Y_i]}\left[E[Y_iT_i] + \frac{E[Y_i^2]}{2}\right].
\label{eqn:age}
\end{align}
Now consider the packet transmissions that constitute the intervals $T_{i,k-1}$ and $Y_{i,k}$. See Figure~\ref{fig:age} for an illustration. The interval $T_{i,k-1}$ begins with the sensor $i$ being polled by the gateway and ends with the gateway transmitting to the monitor. Therefore, the interval contains exactly one transmission of length distributed as $f_{X_i}$ and one distributed as $f_{X_0}$. It may contain zero or more transmissions by sensors other than $i$. The interval $Y_{i,k}$ must contain at least one transmission by sensor $i$ and at least one to the monitor. It can contain zero or more of the other sensors.

Let $s_{j,k}$ be the number of transmissions by sensor $1 \le j \le n$ in the interval $(b_{i,k-1}, b_{i,k})$ of length $Y_{i,k}$. We have $s_{i,k} \ge 1$ and $s_{j,k} \ge 0$ for $j \ne i$. Let $s_{0,k}$ be the number of transmissions to the monitor. Similarly, let $s'_{j,k}$ and $s'_{0,k}$ respectively be the number of transmissions by sensor $j$ and by the gateway to the monitor, over the interval $T_{i,k-1}$. We have, $s'_{j,k} \ge 0$ for $j\ne i$, $s'_{i,k} = 1$, and $s'_{0,k} = 1$. Let $X_{j,l}$ and $X'_{j,l}$, $1\le j \le n, l\ge 0$, be the lengths of the $l$\textsuperscript{th} transmissions of sensor $j$ in the intervals $Y_{i,k}$ and $T_{i,k-1}$, respectively. Similarly, define $X_{0,l}$ and $X'_{0,l}$ for transmissions to the monitor. We can write
\begin{align}
Y_{i,k} &= \sum_{j=1}^{n} \sum_{l=1}^{s_{j,k}} X_{j,l} + \sum_{l=1}^{s_{0,k}} X_{0,l},\nonumber\\
T_{i,k-1} &= X'_{i,1} + X'_{0,1} + \sum_{\substack{j\ne i\\ j\ge 1}} \sum_{l=1}^{s'_{j,k}} X'_{j,l}.
\end{align}
Observe that increasing the frequency of transmissions to the monitor for a given polling schedule of sensors by the gateway will on an average reduce the age of an update that is received by the monitor.  This is because the average time between polling a sensor and sending its update to the monitor will reduce. On an average $T_{i,k}$ will be smaller. However, more frequent transmissions to the monitor will mean more transmissions to the monitor that don't have the latest polled update of a sensor. On an average one would expect a larger $Y_{i,k}$, which is undesirable for a small sensor age at the monitor. Next, we will derive a policy that captures this trade-off between the intervals $T$ and $Y$. 
\section{Poll-$s$ Then Send To Monitor}
\label{sec:policy}
We consider the set $M^{(s)}$ of policies, parameterized by $1 \le s\le n$, in which the gateway transmits polled sensor updates to the monitor every time after a fixed number $s$ of sensors (not necessarily unique) have been polled by it. We will begin by considering the case of homogeneous transmission lengths that is when sensors' polling times are identical and independent. Later we will briefly consider the heterogeneous setting in which polling times may not be identical.

\subsection{Homogeneous Transmission Lengths}
Next we state the optimal method of polling the sensors before transmitting to the monitor, for any given $s\ge 1$ and for when all the sensors have iid transmission times.

\begin{theorem}
Assume the transmission times of all sensors are iid. We have $X_i \sim f_X(x), 1\le i\le n$. Transmissions to the monitor are independent and distributed as $f_{X_0}(x)$. For any given $s\ge 1$, at a decision instant $t$ at which the policy must poll a sensor, polling a sensor $i$ whose age $\ageGW_i(t)$ at the gateway is the largest, achieves the minimum AoI~(\ref{eqn:AoI}) among the policies in the set $M^{(s)}$.
\label{thm:MAF}
\end{theorem}
\begin{corollary}
	Given sensor transmission times are iid, a policy $\mu \in M^{(s)}$ that polls a sensor $i$ that has the largest age $\ageGW_i(t)$ must poll all other sensors exactly once before it polls the sensor $i$ again.
	\label{corr:rr}
\end{corollary}
For the proofs, please see the appendix. We will now calculate the average age~(\ref{eqn:avgAgei}) of sensor $i$ when an optimal policy $\mu \in M^{(s)}$ is used for any given $s\ge 1$. To do so, we must calculate $E[Y_{i,k}]$, $E[Y_{i,k} T_{i,k-1}]$, and $E[Y^2_{i,k}]$. Note that there can be $0 \le R \le s-1$ other sensors that are polled after sensor $i$ is polled and before the polled updates are transmitted to the monitor. To calculate the expected values above, we will first calculate the conditional expectations conditioned on the knowledge of $R$.

Suppose $R$ sensors were polled after sensor $i$ and before sending the polled updates to the monitor that are received at $b_{i,k-1}$. Given Corollary~\ref{corr:rr}, after the transmission to the monitor, we will have $n-R-1$ other sensors that will be polled before sensor $i$ is polled again. Including the transmission of sensor $i$, we will have a total of $n-R$ sensor transmissions. However, since we must transmit to the monitor after every $s$ sensor transmissions, the $n-R$ sensor transmissions will see $L \triangleq \lceil \frac{n-R}{s}\rceil$ ($\lceil x \rceil$ is the ceiling function) transmissions to the monitor, where the last of these transmissions is when the new update from $i$ will be sent by the gateway to the monitor, which will be received by the monitor at $b_{i,k}$. 

The total number of sensors that will be polled in the interval $(b_{i,k-1}, b_{i,k})$ of length $Y_{i,k}$ is therefore $Ls$. The interval $T_{i,k-1}$ has $R+1$ sensor transmissions and $1$ to the monitor. As the sensors have iid transmission times, one would expect $R$ to take values from the set $\{0,1,\ldots,s-1\}$ with equal probability. We can write, using the steady state equivalents $Y_i$ and $T_i$, respectively, of the random variables $Y_{i,k}$ and $T_{i,k-1}$
\begin{subequations}
\begin{align}
&\E{Y_{i}|R} = Ls \E{X} + L \E{X_0},\\
&\E{Y_{i}^2|R} = Ls \Var{X} + L \Var{X_0} + 2 L^2 s \E{X} \E{X_0}\nonumber\\ 
&\qquad\qquad\quad+  L^2 s^2 (\E{X})^2 + L^2 (\E{X_0})^2,\\
&\E{Y_{i} T_{i}|R} = L (s + R + 1) \E{X} \E{X_0}\nonumber\\
&\qquad\qquad\quad\qquad+ Ls (R+1) (\E{X})^2 + L (\E{X_0})^2.
\end{align}
\end{subequations}
Averaging over $R$ and substituting $\eta_1 \triangleq \E{X_0}/\E{X}$ and $\eta_2 \triangleq \Var{X_0}/\Var{X}$ gives
\begin{subequations}
	\begin{align}
	&\E{Y_{i}} = \E{L}\E{X} (s+\eta_1),\\
	&\E{Y_{i}^2} = \E{L}\Var{X} (s + \eta_2)\nonumber\\ 
	&\qquad\qquad+ \E{L^2} (\E{X})^2 (s+\eta_1)^2,\\
	&\E{Y_{i} T_{i}} = (s+\eta_1)(\E{X})^2 (\E{LR} + (\eta_1 + 1)\E{L}).
	\end{align}
\end{subequations}
Substituting in Equation~(\ref{eqn:age}) we get
\begin{align}
\ageMAvg_i &= \frac{\Var{X}}{2\E{X}} \frac{(s+\eta_2)}{(s+\eta_1)} + \frac{\E{L^2}}{2\E{L}}\E{X}(s+\eta_1)\nonumber\\
&\qquad\qquad + \frac{\E{LR}}{\E{L}} \E{X} + (\eta_1 + 1)\E{X},
\label{eqn:avgAgeHomo}
\end{align}
for any sensor $1\le i\le n$. Substituting in~(\ref{eqn:AoI}) gives us the AoI for a policy that polls a sensor using the maximum age first (MAF) rule and has the gateway transmit to the monitor after polling exactly $s$ sensors. Let $s=s^*$ minimize the AoI. We refer to the resulting policy as $s^*,$MAF.

To gain insight into how $s^*$ changes with $n$, we consider the approximation $L\approx n/s$. This allows us to rewrite $\E{LR} = (n/s) \E{R}$. Further note that $\E{R} = (s-1)/2$. Making these substitutions in~(\ref{eqn:avgAgeHomo}) for $\ageMAvg_i$  we get an approximation 
\begin{align}
\widehat{\ageMAvg}_i &= \frac{\Var{X}}{2\E{X}}\frac{(s+\eta_2)}{(s+\eta_1)}\nonumber\\
 &\quad\quad+ \E{X} \left(\frac{n}{2s} (s+\eta_1) + \frac{s-1}{2} + (\eta_1 + 1)\right).
\end{align}
Now consider the case when $\Var{X} = 0$ (sensor transmission times are deterministic) or $\eta_1 = \eta_2$. Solving for the first-order necessary condition for optimality, it is easy to show that $\widehat{\ageMAvg}_i$ is minimized at $s = \sqrt{\eta_1 n}$. Interestingly, the optimal number of sensors to poll before transmitting to the monitor increases as the square root of the number of sensors in the network. This observation motivates a heuristic policy $\widehat{s^*},\text{MAF}$ that sets $s = \widehat{s^*} \triangleq \text{round}(\sqrt{\eta_1 n})$.
\subsection{Heterogeneous Transmission Lengths} The sensors have transmission lengths that like before are mutually independent but may have non-identical distributions. We restrict ourselves to policies in the set $M^{(s)}$. We consider a policy that polls a sensor that minimizes the expected change in the average age that results from the sensor's transmission. Suppose a sensor $i$ is polled at time $t$. The expected change in the average age at the gateway of each of the other sensors is $\E{X_i}$. The expected change in the age of sensor $i$ is $\E{X_i} - \ageGW_i(t)$. The change in age averaged over sensors is $\frac{n\E{X_i} - \ageGW_i(t)}{n}$.

The policy polls sensor $i^*$ that minimizes the change in age. We will call this as the \emph{minimum change} in \emph{age first} (MCA) rule. We have
\begin{align}
i^* = \argmin_i \E{X_i} - \ageGW_i(t)/n.
\label{eqn:minCost}
\end{align}
Note that this policy polls the sensor with the maximum age first when the $X_i$ are iid. The MCA rule combined with $s^*$ and $\widehat{s^*}$ gives, respectively, the policies $s^*,\text{MCA}$ and $\widehat{s^*},\text{MCA}$.
\section{Evaluation and Conclusions}
\label{sec:eval}
\begin{figure*}
	\begin{center}
		\subfloat[]{\includegraphics[width = 0.24\linewidth]{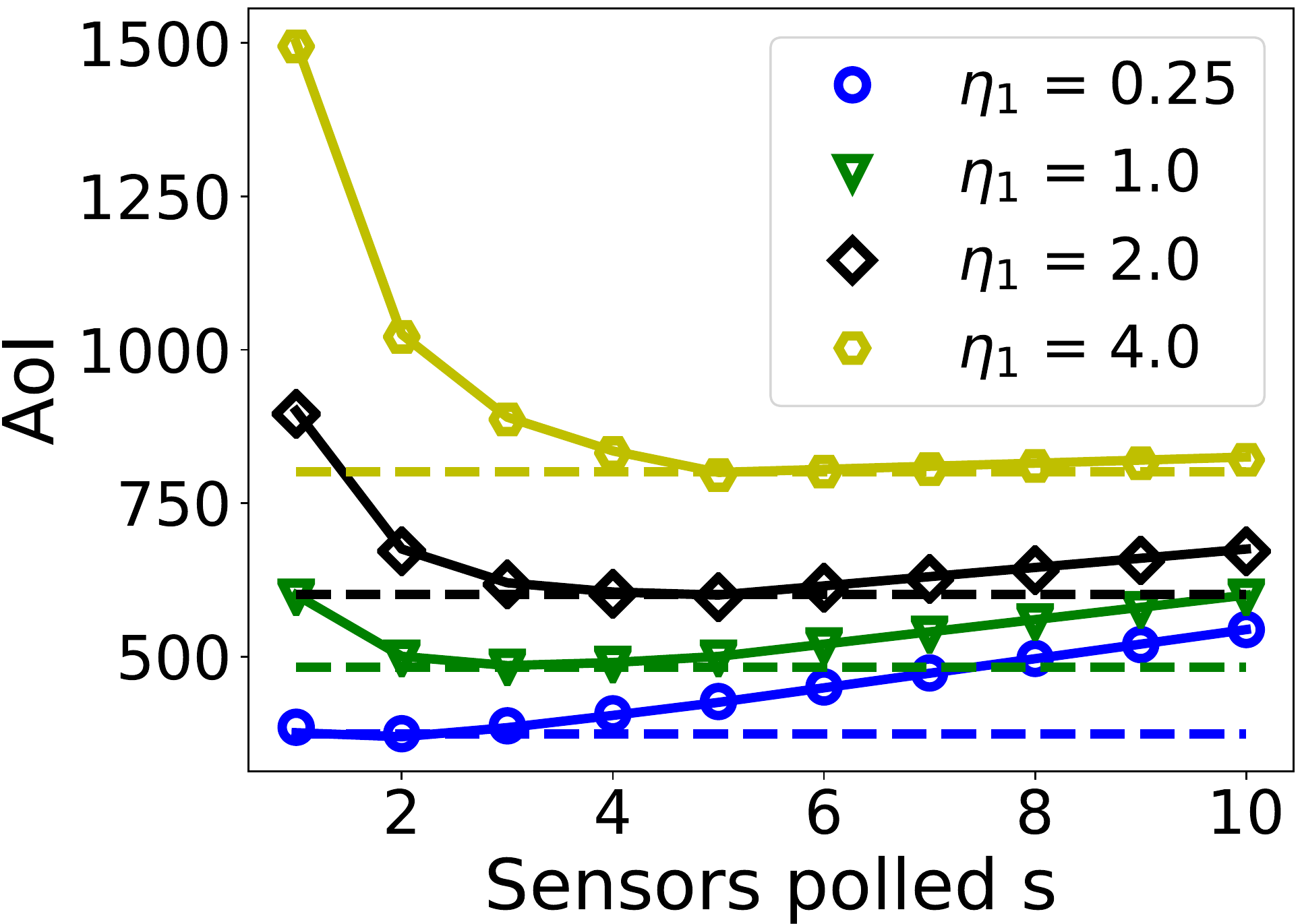}
	\label{fig:avg_age_vs_s_TN}	
}
\subfloat[]{\includegraphics[width = 0.24\linewidth]{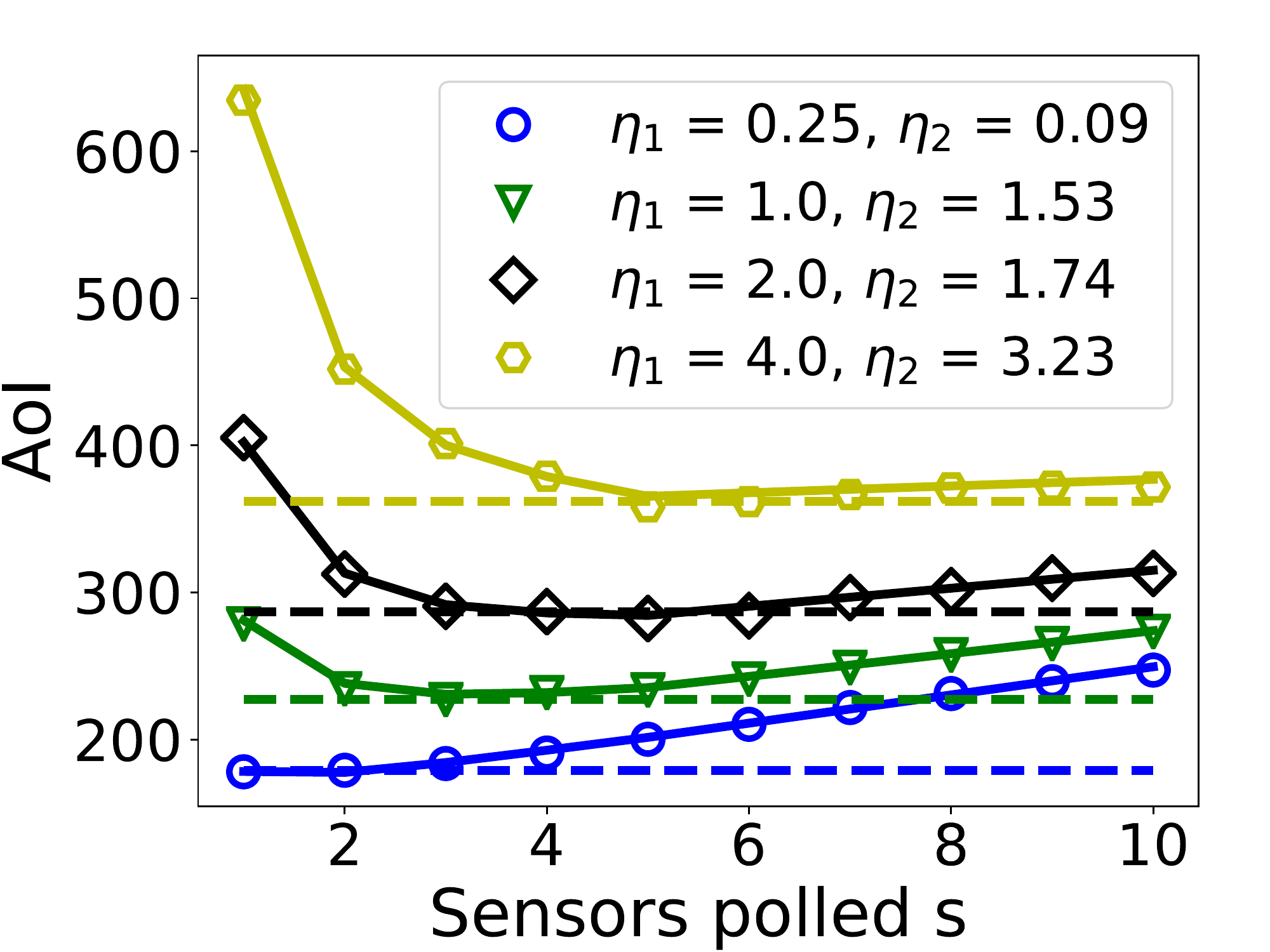}
	\label{fig:avg_age_vs_s_HE}	
}
\subfloat[]{\includegraphics[width = 0.24\linewidth]{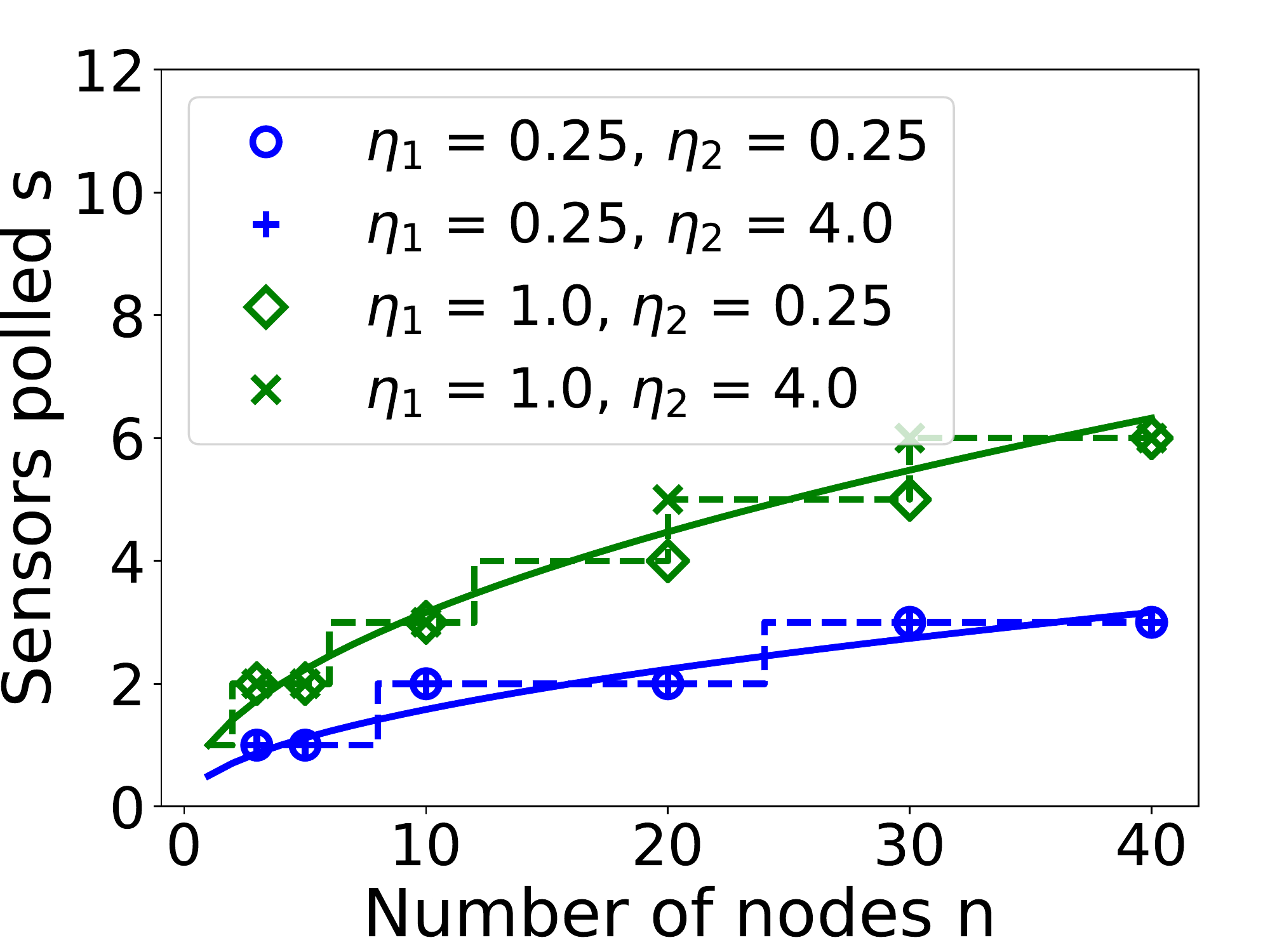}
	\label{fig:s_vs_n_TN}	
}
\subfloat[]{\includegraphics[width = 0.24\linewidth]{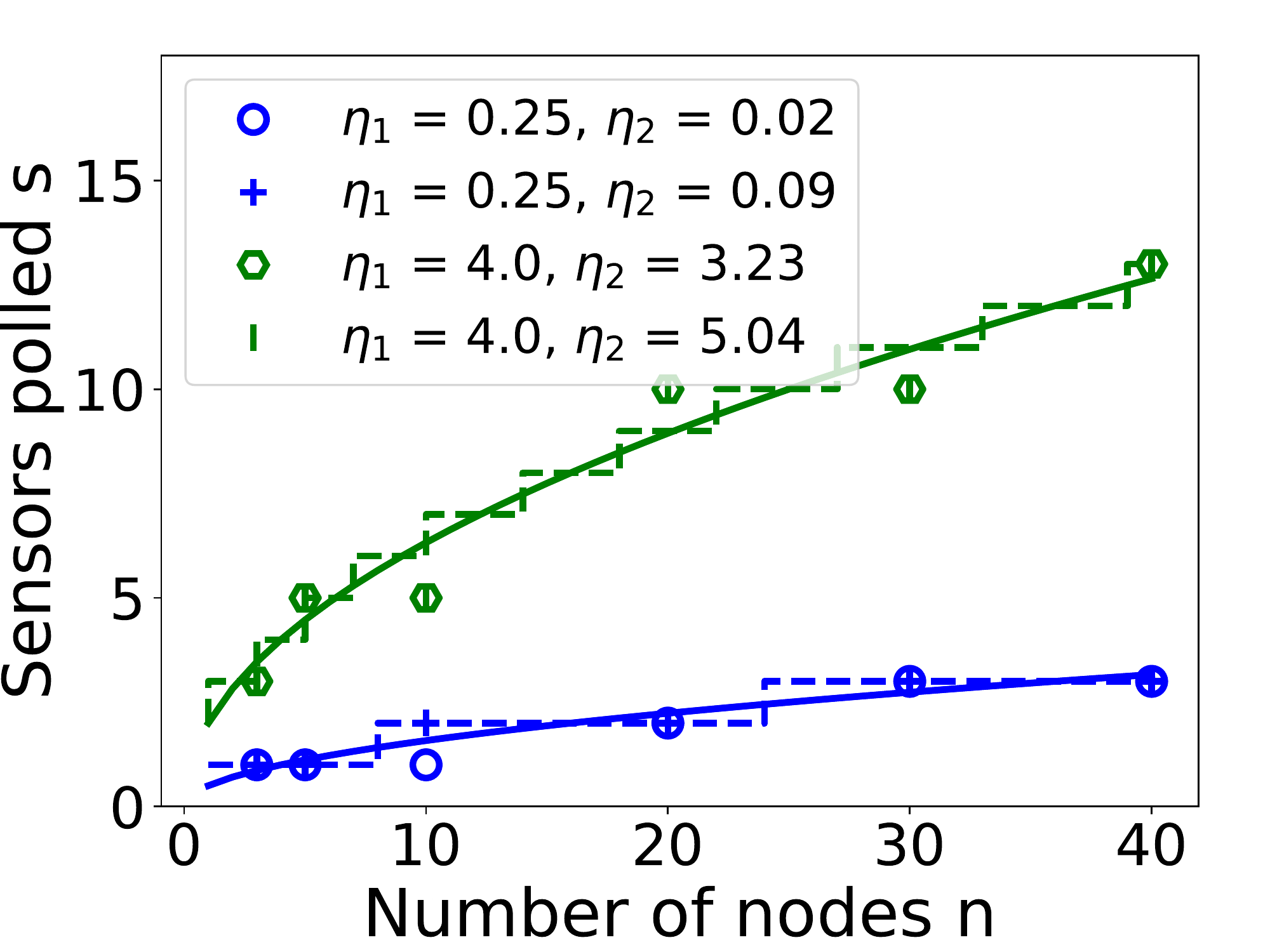}
	\label{fig:s_vs_n_HE}	
}
		\caption{\small \textbf{(a)} Transmission times were simulated using a truncated Gaussian (range space $(0,\infty)$). For sensors, the mean was $50$ and variance $25$. $\eta_2$ was set to $4$. The horizontal dashed lines show the age obtained using $\widehat{s^*}$,MAF. The markers show AoI obtained via simulation. The line shows the AoI obtained analytically. \textbf{(b)} Instead of the truncated Gaussian in Figure (a) we used a hyperexponential with mean $20$ and variance $1300$. \textbf{(c)} Transmission times were as in (a). The solid lines are $\sqrt{\eta_1 n}$. The dashed staircase shows $s=\widehat{s^*}$. The corresponding markers are $s^*$. \textbf{(d)} Same as (c) but for the hyperexponential used for (b).}	
	\end{center}
	\label{fig:etan_plots}
	\vspace{-0.33in}
\end{figure*}

\begin{figure*}
	\begin{center}
		\subfloat[]{\includegraphics[width = 0.24\linewidth]{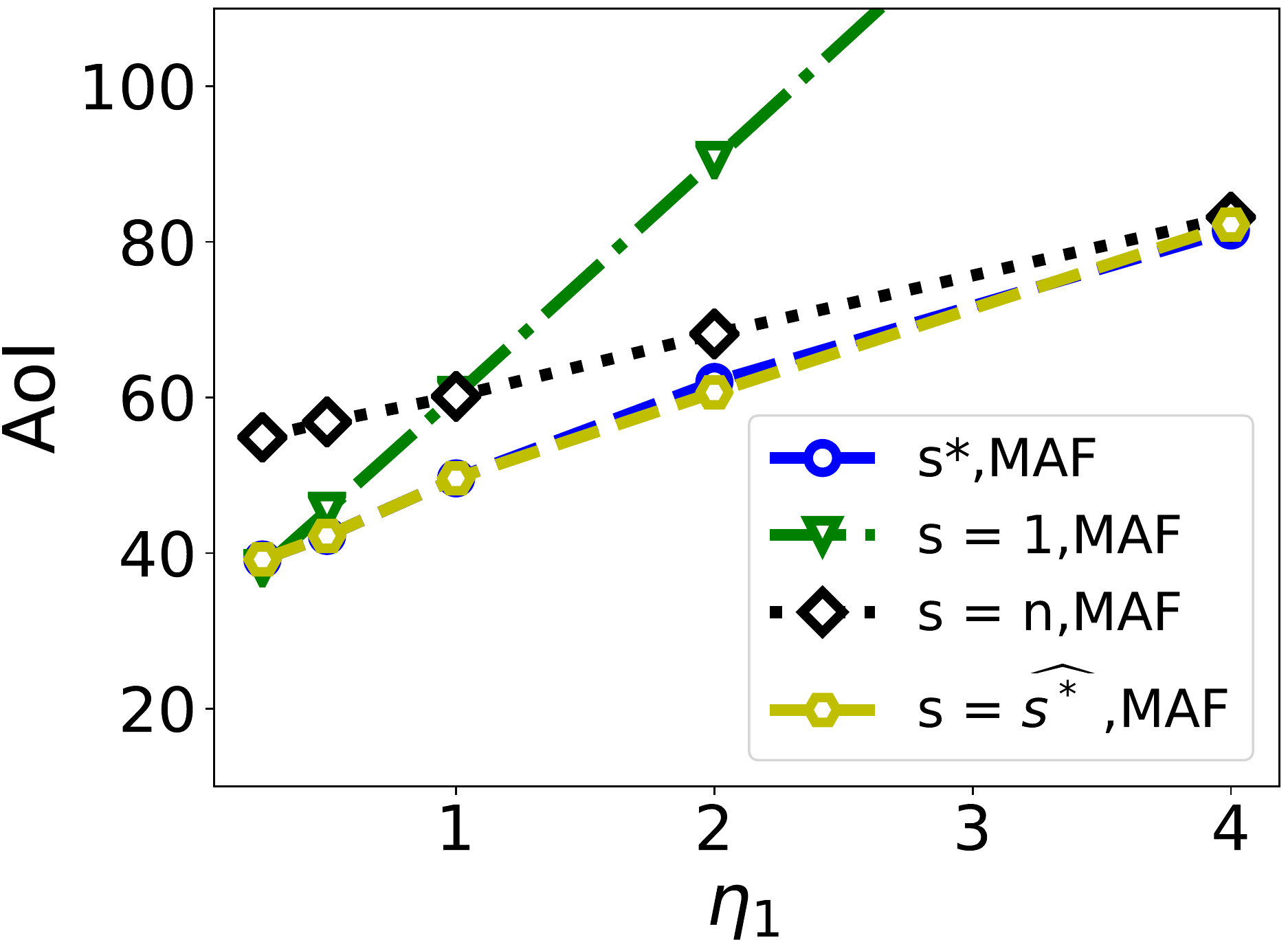}
			\label{fig:avg_age_homo_unif}	
		}
		\subfloat[]{\includegraphics[width = 0.24\linewidth]{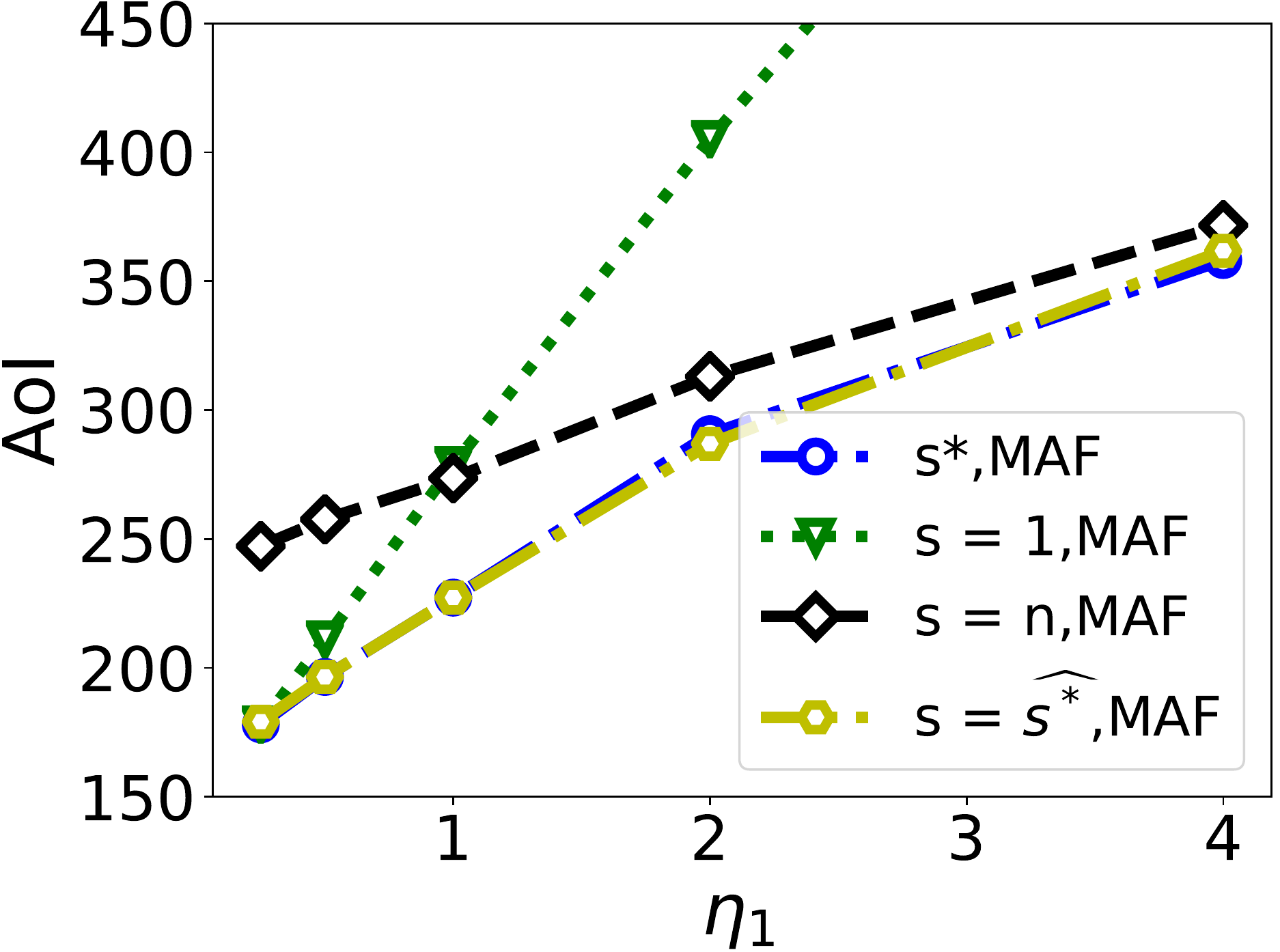}
			\label{fig:avg_age_homo_hyperexp}	
		}
		\subfloat[]{\includegraphics[width = 0.24\linewidth]{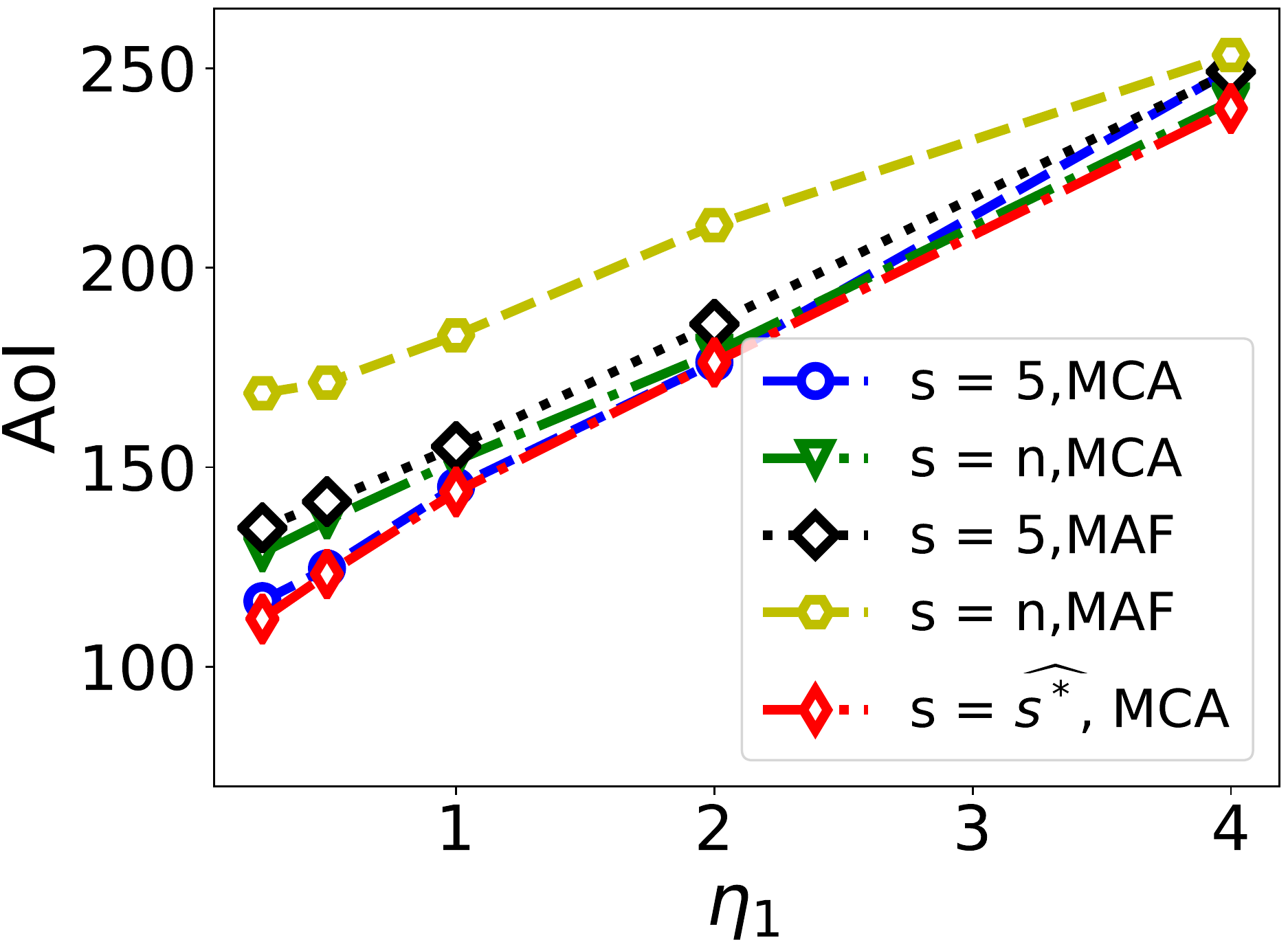}
			\label{fig:avg_age_hetero_exp}	
		}
		\subfloat[]{\includegraphics[width = 0.24\linewidth]{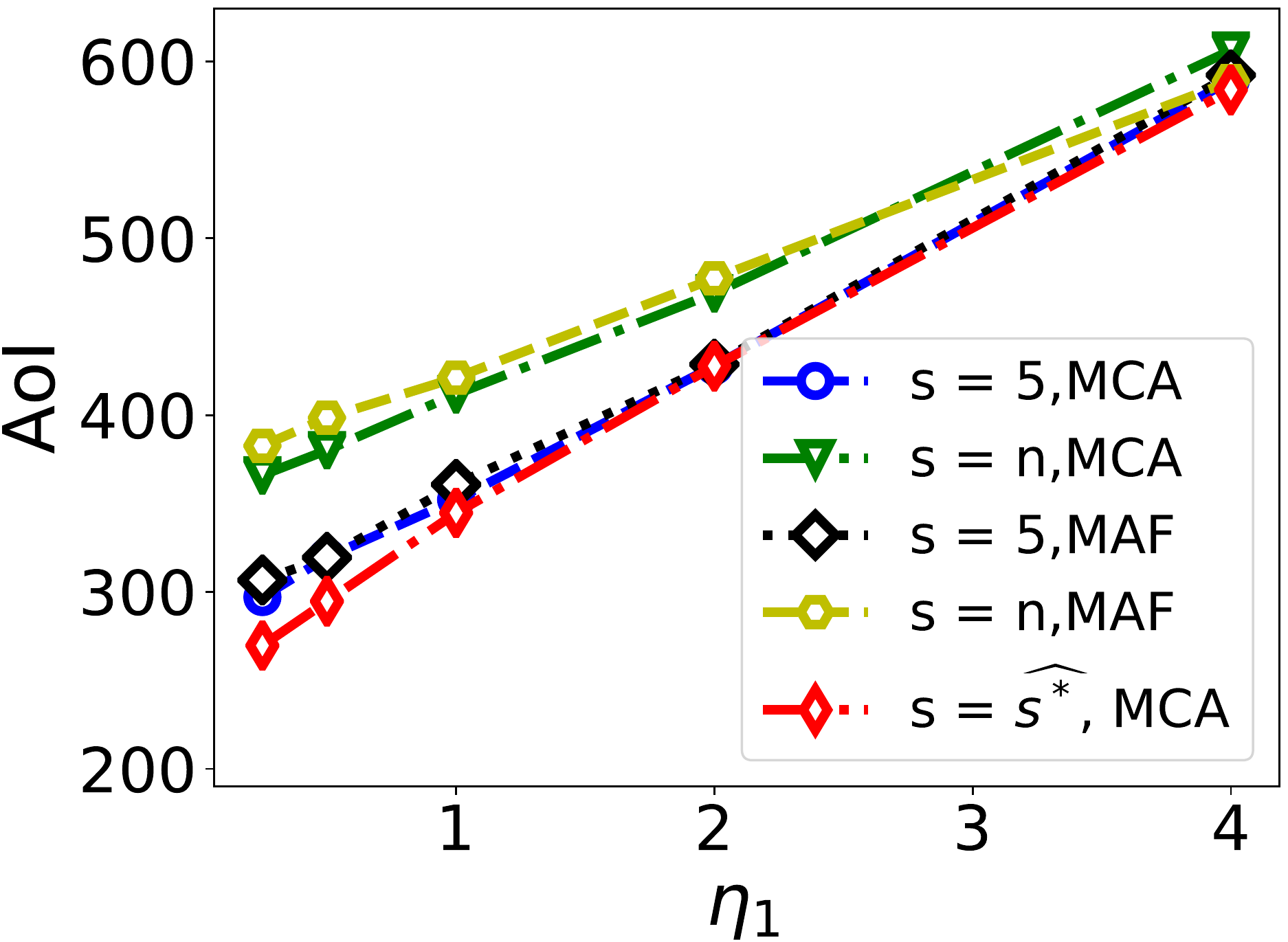}
			\label{fig:avg_age_hetero_hyperexp}	
		}
		\caption{\small Average age obtained for $n=10$ sensors. \textbf{(a)} Transmission times are uniformly distributed. For sensors we chose values over the range $(0,10)$ in an iid manner. \textbf{(b)} Transmission times are hyperexponential. For sensors, the mean was $20$ and the variance was $1300$. $(\eta_1,\eta_2)$ were chosen as $(0.25, 0.06)$, $(0.5, 0.25)$, $(1,1)$, $(2,1.6)$, and $(4,3.2)$. \textbf{(c)} Transmission times are exponentially distributed. However, non-identical for the sensors. \textbf{(d)} As in (c), but with hyperexponential transmission times. $(\eta_1,\eta_2)$ were chosen as $(0.25, 0.07)$, $(0.5, 0.3)$, $(1,1.2)$, $(2,4.6)$, and $(4,18.6)$.}	
	\end{center}
	\vspace{-0.3in}
\end{figure*}
We conclude with a representative sample of our evaluation of the proposed policies. We simulated a wide selection of transmission time distributions, including that of uniform, exponential, truncated Gaussian, and hyperexponential (coefficient of variation greater than $1$). Figures~\ref{fig:avg_age_vs_s_TN} and~\ref{fig:avg_age_vs_s_HE} demonstrate how AoI varies as a function of $s$ for a network of $n=10$ sensors. Observe that the choice of number of sensors the gateway must poll before transmitting to the monitor is crucial to minimizing AoI. Also, as $\eta_1$ increases, that is the average transmission time to the monitor becomes larger in comparison to the average sensor transmission time, the optimal choice moves toward the total number of sensors.

Figures~\ref{fig:s_vs_n_TN} and~\ref{fig:s_vs_n_HE} show the optimal $s^*$ and the approximation $\widehat{s^*}$ as a function of the number of nodes. The choice of transmission time distributions for the figures are same as of, respectively, Figures~\ref{fig:avg_age_vs_s_TN} and~\ref{fig:avg_age_vs_s_HE}. It is clear that the optimal values follow $\sqrt{\eta_1 n}$ closely. Also, for any selection of $\eta_1$, the optimal doesn't vary much over different selections of $\eta_2$.

Figures~\ref{fig:avg_age_homo_unif} and~\ref{fig:avg_age_homo_hyperexp} compare the performance of $s^*$ and $\widehat{s^*}$ with the choices of $s=1$ and $s=n=10$ for when the sensor transmission times are iid (homogeneous setting). For all $s$, sensors are selected in a MAF manner. Not surprisingly, $s^*$ does much better than the choices of $s=1$ and $s=n=10$. However, it is worth noting that $\widehat{s^*}$ does almost as well as $s^*$. 

We conclude with Figures~\ref{fig:avg_age_hetero_exp} and~\ref{fig:avg_age_hetero_hyperexp} that compare the MCA with MAF for different selections of $s$ for when the sensor transmission times are independent but non-identical (heterogeneous setting). We selected the mean for each sensor uniformly and randomly from $(0,40)$. The mean of the transmission to gateway is obtained as the average of means selected for the sensors scaled by $\eta_1$. To keep the presentation simple, we only state the averages of the variances of the sensor transmission times, which were $\approx 200$ and $\approx 1600$, respectively, for Figures~\ref{fig:avg_age_hetero_exp} and~\ref{fig:avg_age_hetero_hyperexp}. In both figures, $\widehat{s^*},$MCA does the best. In Figure~\ref{fig:avg_age_hetero_exp}, in which the transmission times were exponentially distributed, MAF does much worse than MCA. However, in Figure~\ref{fig:avg_age_hetero_hyperexp}, in which the transmission times were hyperexponential, MAF does almost as well as MCA. For both the cases, other choices of $s$ do worse than $\widehat{s^*}$.

\begin{spacing}{1.0}
\bibliographystyle{IEEEtran}
\bibliography{main}
\end{spacing}

\appendices
\section{Proof of Theorem~\ref{thm:MAF}}
\label{app:thmProof}
We will broadly follow the outline of the proof of Theorem $3.2$ in~\cite{AO-STS-MSS}. Let $\mu\in M^{(s)}$ be a policy, which at a decision instant at which a sensor must be polled, always a polls a sensor that has the maximum age first. Let $\pi$ be any policy in the set of policies $M^{(s)}$. A policy influences the age vectors at the gateway and the monitor.

We will use $[i]$ to refer to the $i$\textsuperscript{th} largest age. Let $\widehat{\ageMVec}_\pi(t) = \vec{\ageM_{\pi,[1]}(t)&\cdots&\ageM_{\pi,[n]}(t)}$ and $\widehat{\ageGWVec}_\pi(t) = \vec{\ageGW_{\pi,[1]}(t)&\cdots&\ageGW_{\pi,[n]}(t)}$ be the vectors, respectively of the ordered sensor ages at the monitor and the gateway, when using $\pi$. The corresponding age process is $\{(\widehat{\ageMVec}_\pi(t), \widehat{\ageGWVec}_\pi(t)), t\ge 0\}$. We will assume that $\widehat{\ageMVec}_\pi(0^-) = \widehat{\ageGWVec}_\pi(0^-) = \widehat{\ageMVec}_\mu(0^-) = \widehat{\ageGWVec}_\mu(0^-)$. The following lemma compares the policies $\mu$ and $\pi$.
\begin{lemma}
	Let $\widehat{\ageMVec}_\pi(0^-) = \widehat{\ageGWVec}_\pi(0^-) = \widehat{\ageMVec}_\mu(0^-) = \widehat{\ageGWVec}_\mu(0^-)$ for all policies $\pi\in M^{(s)}$. We will also assume that both the policies begin at $t=0$ by either first polling a sensor or having the gateway transmit to the monitor. For any given $s$, we have $\{\widehat{\ageMVec}_\mu(t), t\ge 0\} \le_{st} \{\widehat{\ageMVec}_\pi(t), t\ge 0\}$.
\label{lem:stochasticOrdering}	
\end{lemma}
As in~\cite{AO-STS-MSS}, we define coupled processes $\ageGWVecCoupled_\pi(t)$ and $\ageGWVecCoupled_\mu(t)$ that are governed by the stochastic laws of, respectively, the processes $\widehat{\ageGWVec}_\pi(t)$ and $\widehat{\ageGWVec}_\mu(t)$. They are coupled in that they observe the same randomly drawn sample of transmission times at every decision time instant. The processes can be coupled because sensor transmission times are iid, independent of the transmission time to the monitor, and because both the policies $\pi$ and $\mu$ must have the gateway transmit to the monitor after polling exactly $s$ sensors. Finally, note that we assume that both policies choose either to poll a sensor (of their choice) or to have the gateway transmit to the monitor at time $t=0$. Also define $\ageMVecCoupled_\pi(t)$ and $\ageMVecCoupled_\mu(t)$, which are the age processes at the monitor corresponding to, respectively, $\ageGWVecCoupled_\pi(t)$ and $\ageGWVecCoupled_\mu(t)$.

To prove lemma~\ref{lem:stochasticOrdering} it is sufficient~\cite{AO-STS-MSS} to show that
\begin{align}
P[\ageMVecCoupled_\mu(t) \le \ageMVecCoupled_\pi(t), t\ge 0] = 1.
\label{eqn:stochOrderingProb1}
\end{align}

We will need the following lemma to show the above. 

\begin{lemma}
Consider a certain time $t$ at which the policies must choose an action. Note that for the coupled processes the decision making times coincide. Let $t'$ be the time when the transmission due to the decision completes. If $\ageGWCoupled_{\mu,[i]}(t) \le \ageGWCoupled_{\pi,[i]}(t)$ for $1\le i \le n$, then $\ageGWCoupled_{\mu,[i]}(t') \le \ageGWCoupled_{\pi,[i]}(t')$.
\label{lem:induction}
\end{lemma}
\begin{proof}
Note that, given the coupled processes, both policies must either choose to poll a sensor or to send to the monitor at $t$. First consider the case that they choose to poll. 

Let the resulting transmission time be $z$. For the policy $\mu$, which polls the sensor with the maximum age,
\begin{subequations}\label{eqn:coupledGWEvolmu1}
\begin{align}
\ageGWCoupled_{\mu,[n]}(t') &= z,\\
\ageGWCoupled_{\mu,[i]}(t') &= \ageGWCoupled_{\mu,[i+1]}(t) + z,\  1\le i \le n-1.
\end{align}
\end{subequations}
Suppose the policy $\pi$ polls the sensor that has the $j$\textsuperscript{th} largest age at time $t$.
\begin{subequations}
\begin{align}
\ageGWCoupled_{\pi,[n]}(t') &= z,\\
\ageGWCoupled_{\pi,[i]}(t') &= \ageGWCoupled_{\pi,[i]}(t) + z,\  1\le i \le j-1,\\
\ageGWCoupled_{\pi,[i]}(t') &= \ageGWCoupled_{\pi,[i+1]}(t) + z,\  j\le i \le n-1.
\end{align}
\end{subequations}

Since by assumption $\ageGWCoupled_{\pi,[i]}(t) \ge \ageGWCoupled_{\mu,[i]}(t)$, $\forall i$, and given that all sensors other than those polled see an increase in their age at the gateway by $z$, the above update of ages implies that $\ageGWCoupled_{\pi,[i]}(t') \ge \ageGWCoupled_{\mu,[i]}(t')$, for all $i$.

Instead, if the policies choose to send to the monitor, all ages increase by $z$. Age orderings at $t$ are maintained at $t'$. This completes the proof of the lemma.
\end{proof}
\begin{proof}[Proof of Lemma~\ref{lem:stochasticOrdering}]
	Note that at times other than the instants when the monitor receives a transmission the ages at the monitor $\ageMCoupled_{\mu,i}(t)$ and $\ageMCoupled_{\pi,i}(t)$ increase in time with slope $1$. These are reset to the corresponding gateway ages (see~(\ref{eqn:ageEvolAtMonitor})) when the monitor receives a transmission from the gateway. Lemma~\ref{lem:induction} implies that at the end of all decision instants, given that the coupled processes $\ageGWVecCoupled_\pi(t)$ and $\ageGWVecCoupled_\mu(t)$ have the same age vectors at $t=0^{-}$, $\ageGWCoupled_{\mu,[i]}(t) \le \ageGWCoupled_{\pi,[i]}(t)$, for all $i$. Thus, on reset the ages at the monitor too satisfy $\ageMCoupled_{\mu,[i]}(t) \le \ageMCoupled_{\pi,[i]}(t)$. This implies that~(\ref{eqn:stochOrderingProb1}) is true, which proves the lemma.
\end{proof}
\begin{proof}[Proof of Theorem~\ref{thm:MAF}]
Since the AoI~(\ref{eqn:AoI}) is the expectation of a non-decreasing functional of the age process $\{\widehat{\ageMVec}(t), t\ge 0\}$, Lemma~\ref{lem:stochasticOrdering} implies that the policy $\mu$ leads to the smallest AoI among policies $\pi \in M^{(s)}$~\cite{shaked2007stochastic}. This proves the theorem.
\end{proof}
\section{Proof of Corollary~\ref{corr:rr}}
\label{app:corrProof}
\begin{proof}[Proof of Corollary~\ref{corr:rr}]
Observe from~(\ref{eqn:coupledGWEvolmu1}) how the sensor ages are updated at the gateway after a sensor with the largest age is polled and its transmission is received. Specifically, the sensor that is polled sees its age reset to the time its transmission took. The ages of all other sensors increase by this transmission time and maintain their ordering with respect to each other. The polled sensor ends up with the smallest age. There are $n-1$ sensors that have an age larger than it. All these $n-1$ must be polled before the sensor is polled again. This proves the corollary.	
\end{proof}

\end{document}